\documentclass[runningheads,envcountsame]{llncs}
\usepackage{common/prelude}
\usepackage[protrusion=true,expansion=false]{microtype}
\usepackage{hyperref}%
\usepackage{cleveref}
\usepackage{doi}%
\usepackage{nicefrac}
\usepackage{standalone}
\usepackage{xcolor}
\usepackage{makecell} 
\usepackage{tabularx} 
\usepackage{bbm}

\setcellgapes{2pt}
\usetikzlibrary{backgrounds,positioning,arrows.meta,calc}

\undef\todo
\usepackage{verbatim}
\usepackage{soul} 
\soulregister\cite7
\usepackage[commandnameprefix=ifneeded,final]{changes}
\definechangesauthor[color=orange]{yan}
\definechangesauthor[color=red]{cris}
\sethighlightmarkup{\IfIsColored{{\sethlcolor{authorcolor!30}\hl{#1}}}{#1}}

\title{Sub-polynomial parameterized complexity of~\texorpdfstring{$k$-core}{k-core}}
\author{Yan S. Couto\inst{1}\orcidID{0009-0005-5850-8696} \and
Cristina G. Fernandes\inst{1}\orcidID{0000-0002-5259-2859}}
\authorrunning{Y.\,S. Couto and C.\,G. Fernandes}
\institute{University of São Paulo, São Paulo, SP, Brazil\\
\email{\{yancouto,cris\}@ime.usp.br}}
\makeatletter
\newcommand{\setciteauthor}[2]{\@namedef{manualcite@#1}{#2}}

\NewDocumentCommand{\citeauthor}{m}{%
  \ifcsname manualcite@#1\endcsname
    \@nameuse{manualcite@#1}%
  \else
    \@latex@warning{Author undefined for cite key '#1'}%
    \textcolor{red}{\textbf{[Missing Author]}}%
  \fi
}

\NewDocumentCommand{\citet}{o m}{%
  \citeauthor{#2}~%
  \IfValueTF{#1}{%
    \cite[#1]{#2}%
  }{%
    \cite{#2}%
  }%
}
\makeatother

\setciteauthor{anderson_p-complete_1984}{Anderson and Mayr}
\setciteauthor{courcelle_graph_1990}{Courcelle}
\setciteauthor{bannach_parallel_2017}{Bannach and Tantau}
\setciteauthor{vollmer_introduction_1999}{Vollmer}
\setciteauthor{elberfeld_logspace_2010}{Elberfeld, Jakoby and Tantau}
\setciteauthor{elberfeld_space_2015}{Elberfeld, Stockhusen and Tantau}
\setciteauthor{etessami_counting_1997}{Etessami}
\setciteauthor{couto_hardness_2025_full}{Couto and Fernandes}
\setciteauthor{wang_fast_2023}{Wang et al.}
\setciteauthor{bodlaender_contraction_2006}{Bodlaender, Wolle and Koster}

\renewcommand{\problem}[1]{\textup{\textsc{#1}}\xspace}
\newcommand{\kcore}{$k$-core\xspace}

\newcommand{\Kcore}{\problem{$k$-core}}
\newcommand{\NEkcore}{\problem{NE-$k$-core}}
\newcommand{\Tcore}{\problem{$3$-core}}

\newcommand{\ORD}{\problem{ORD}}
\newcommand{\pA}{\problem{A}}

\newcommand{\cclass}[1]{\textbf{\textup{#1}}\xspace}
\newcommand{\para}[1]{\cclass{\text{para-}#1}}
\newcommand{\NP}{\cclass{NP}}
\renewcommand{\P}{\cclass{P}}
\newcommand{\NC}{\cclass{NC}}
\newcommand{\FNC}{\cclass{FNC}}
\renewcommand{\L}{\cclass{L}}
\newcommand{\NL}{\cclass{NL}}
\newcommand{\XL}{\cclass{XL}}
\newcommand{\paraNC}{\para{\NC}}
\newcommand{\paraNCeps}{\para{\ensuremath{\NC^{2+\epsilon}}}}
\newcommand{\p}{\mathrm{p}}

\renewcommand{\N}{\mathbb{N}}

\newcommand{\CRCW}{\texttt{CRCW}\xspace}

\newcommand{\PRAM}{\texttt{PRAM}\xspace}
\newcommand{\MSO}{MSO\xspace}

\DeclareMathOperator{\tw}{tw}
\DeclareMathOperator{\pw}{pw}

\renewcommand{\emptyset}{\varnothing}
\newcommand{\band}{\wedge}

\newcommand{\bigand}{\bigwedge}

\renewcommand{\mid}{\colon}

\begin{document}

\maketitle

\begin{abstract}
    The \kcore of a graph is its (unique) largest subgraph with minimum degree at least~$k$. For any~$k \geq 3$, deciding whether a given vertex belongs to the~\kcore is a~\P-complete problem, meaning that it is inherently sequential and highly unlikely to admit efficient parallel algorithms, even on graphs of maximum degree~$k+1$. This paper investigates alternative parameterizations of the \kcore problem to identify conditions under which it can be placed into sub-polynomial complexity classes. We prove that the problem is in $\paraNCeps$ when parameterized by treewidth, and in $\paraNC^3$ when parameterized by~$k$ on chordal graphs. Furthermore, we introduce a novel $\NC^{3}$ algorithm for interval graphs when~$k = \Oh(\lg v(G))$, which relies on an improved parameterization by pathwidth. Finally, we establish corresponding lower bounds, demonstrating that, even with these parameterizations, computing the~\kcore remains~\L-hard, meaning it requires at least logarithmic space. These findings explore the boundary of parallel tractability for the~\kcore problem by highlighting the graph parameters that make it inherently sequential.

\keywords{community search \and core decomposition \and parameterized complexity \and parallel algorithms \and interval and chordal graphs}
\end{abstract}

\section{Introduction}

The~\deff{\kcore} of a graph is its (unique) largest subgraph with minimum degree at least~$k$. See~\Cref{fig:kcore} for an example. The problem is a powerful tool for identifying densely connected regions of a graph, with applications across domains such as social network analysis, risk assessment, bioinformatics, and system robustness. See the surveys~\cite{malliaros_core_2020,du_core_2020} for a list of applications.

\begin{figure}[h]
    \centering
    \scalebox{0.7}{\begin{tikzpicture}[
    vertex/.style={circle, draw=black, fill=white, thick, inner sep=2pt, minimum size=7pt},
    core1/.style={fill=green!20, rounded corners=20pt},
    core2/.style={fill=blue!20, rounded corners=15pt},
    core3/.style={fill=red!20, rounded corners=12pt}
]

\node[vertex] (v1) at (0,0)   {H};
\node[vertex] (v2) at (2,0)   {I};
\node[vertex] (v3) at (0,2)   {F};
\node[vertex] (v4) at (2,2)   {G};

\node[vertex] (v0) at (3,1)   {D};
\node[vertex] (v5) at (4,0.5) {E};
\node[vertex] (v6) at (4,1.5) {C};

\node[vertex] (v7) at (5.5,1) {B};
\node[vertex] (v8) at (-1.5,1){A};

\pgfdeclarelayer{background}
\pgfsetlayers{background,main}

\begin{pgfonlayer}{background}
    \fill[core1] ($(v8.west)+(-0.4,0)$) 
              -- ($(v3.north)+(-0.2,0.4)$) 
              -- ($(v4.north)+(0.2,0.4)$) 
              -- ($(v6.north)+(0.4,0.4)$) 
              -- ($(v7.east)+(0.4,0)$) 
              -- ($(v5.south)+(0.4,-0.4)$) 
              -- ($(v2.south)+(0.2,-0.4)$)
              -- ($(v1.south)+(-0.2,-0.4)$)
              -- cycle;

    \fill[core2] ($(v1.south west)+(-0.3,-0.3)$) 
              -- ($(v2.south east)+(0.3,-0.3)$) 
              -- ($(v5.south east)+(0.3,-0.3)$) 
              -- ($(v6.north east)+(0.3,0.3)$) 
              -- ($(v4.north east)+(-0.3,0.3)$) 
              -- ($(v3.north west)+(-0.3,0.3)$) 
              -- cycle;

    \fill[core3] ($(v1.south west)+(-0.2,-0.2)$) 
              -- ($(v2.south east)+(0.2,-0.2)$) 
              -- ($(v4.north east)+(0.2,0.2)$) 
              -- ($(v3.north west)+(-0.2,0.2)$) 
              -- cycle;
\end{pgfonlayer}

\draw[thick] (v1) -- (v2);
\draw[thick] (v1) -- (v3);
\draw[thick] (v1) -- (v4);
\draw[thick] (v2) -- (v3);
\draw[thick] (v2) -- (v4);
\draw[thick] (v3) -- (v4);

\draw[thick] (v4) -- (v6);
\draw[thick] (v5) -- (v6);
\draw[thick] (v0) -- (v5);
\draw[thick] (v0) -- (v6);

\draw[thick] (v6) -- (v7);
\draw[thick] (v1) -- (v8);

\end{tikzpicture}}
    \caption{A graph with its \textcolor{green!80!black}{1-core}, \textcolor{blue}{2-core}, and \textcolor{red}{3-core} highlighted.}
    \label{fig:kcore}
\end{figure}
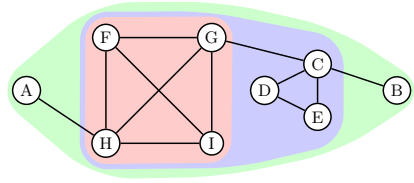

In the sequential setting,~\kcore can be solved efficiently in time proportional to the number of edges~\cite{batagelj_m_2003} by a \emph{peeling} algorithm, in which we repeatedly remove vertices of degree less than~$k$. The efficiency of this algorithm is why this problem attracts attention, because other solutions for densely connected regions of a graph, such as large cliques, are often~\NP-hard even to approximate~\cite{hastad_clique_1999}.

Due to the massive scale of modern networks, significant effort has been dedicated to finding efficient parallel~\cite{liu_parallel_2025} (using multiple cores) and dynamic~\cite{zhang_maintaining_2024} (when edges are added and removed) algorithms for~\kcore. However, the problem is known to be~\P-complete (under~\NC reductions)~\cite{anderson_p-complete_1984}, implying that, if it can be parallelized efficiently, then all other problems in~\P can as well, which is widely believed to be false~\cite[Chapter 5]{greenlaw_limits_1995}. Analogously, it has been shown that an efficient dynamic algorithm for \kcore is unlikely to exist~\cite{couto_hardness_2025}.
The proofs for hardness work for any fixed~$k \geq 3$ on graphs of maximum degree~$k+1$, which means even parameterizing by~$k$ or maximum degree does not make the problem tractable.

To bypass this theoretical bottleneck, we explore the problem through the lens of parameterized complexity. While the most well-known results in parameterized complexity deal with the~\cclass{FPT} (or \para{P}) class, problems like~\kcore are trivially in that class since they are in~\P. Recent works, however, have studied \emph{parallel} tractability~\cite{elberfeld_space_2015,bannach_parallel_2022}, placing problems in parameterized space or parallel classes such as~\para{L} or~\paraNC.
Inspired by these approaches, we study various parameterizations of~\kcore and what they uncover about which graph parameters dictate the sequential nature of the~\kcore problem.

\subsection{Contributions}

Our main contributions explore the limits of the~\kcore problem to identify the sources of its sequential nature, by providing upper and lower bounds for the problem. The results are summarized in~\Cref{tab:summary}. Problem~\problem{$\p_k$-$\pi$} is~$\pi$ parameterized by~$k$. The class~\para{C} is the parameterized version of the decision class~\cclass{C}, while~\cclass{XC} is its slice-wise parameterization (for example, \cclass{XP} allows~$\Oh(|x|^k)$ algorithms while~\para{P} does not). See~\Cref{sec:preliminaries} for the full definitions of problems and complexity classes.

\begin{table}[h]
    \centering
    \crefname{theorem}{Thm.}{Thms.}
    \Crefname{theorem}{Thm.}{Thms.}
    \crefname{corollary}{Cor.}{Cors.}
    \Crefname{corollary}{Cor.}{Cors.}
    \setlength{\tabcolsep}{2.5pt} 
    \resizebox{\textwidth}{!}{%
    \begin{tabular}{|c||c|c|}
        \hline
        Problem & Upper bound & Lower bound \\ \hline 
        \Kcore & \P (\citet{anderson_p-complete_1984}) & \P-hard (\citet{anderson_p-complete_1984}) \\
        \problem{$\p_{\tw}$-$k$-core} & \paraNCeps / \XL (\Cref{res:tw-kcore,res:x-kcore}) & \para{L}-hard (\Cref{res:kcore-para-l-hard}) \\
        \problem{$\p_k$-C-$k$-core} & $\paraNC^3$ (\Cref{res:chordal-nc})  & \para{L}-hard (\Cref{res:kcore-para-l-hard}) \\
        \problem{$\p_k$-I-$k$-core} & \paraNCeps / \XL (\Cref{res:chordal-nc,res:x-kcore})  & \para{L}-hard (\Cref{res:kcore-para-l-hard}) \\
        \problem{I-$\Oh(\lg v(G))$-core} & $\NC^{3}$ (\Cref{res:interval-nc}) & \L-hard (\Cref{res:kcore-l-hard-tw}) \\ \hline
    \end{tabular}}
    \vspace{1mm}
    \caption{Summary of the best-known bounds for the complexity of computing the~\kcore of a graph in general, or in chordal (C) or interval (I) graphs.}
    \label{tab:summary}
\end{table}

\subsection{Related work}

\paragraph{Parallel $k$-core decomposition.}
\citet{anderson_p-complete_1984} showed that the~\kcore problem is~\P-complete, and that every~$(2-\epsilon)$-approximation for it is also~\P-complete, for each~$\epsilon > 0$. In the same article, the authors also provided an~\NC algorithm for a~$(2+\epsilon)$-approximation, and this same technique has been explored for more practical parallel~\cite{liu_parallel_2024}, distributed~\cite{chan_distributed_2019}, and dynamic~\cite{sun_fully_2020} algorithms. There has been more work on practical exact parallel algorithms for~\kcore~\cite{liu_parallel_2025}, but they focus on decreasing contention and optimizing work-efficiency, and their worst-case time complexity remains linear.

\paragraph{Parameterized dense subgraph search.}
Finding a clique of size~$k$, parameterized by~$k$, is a canonical example of intractability: despite the straightforward~$\Oh(v(G)^k)$ algorithm, the problem remains~$\cclass{W}[1]$-complete~\cite{cygan_parameterized_2015}. Indeed, the original definition of~$\cclass{W}[1]$-hardness is formulated around parameterized clique. This has motivated a range of relaxations and alternative parameterizations. For example,~\citet{wang_fast_2023} considered parameterizations of the~$k$-plex problem, which weakens the clique degree requirement. The~\kcore problem is another relaxation of clique, and it is polynomial-time solvable. In this paper, we focus instead on its sub-polynomial complexity.

\paragraph{Parameterized sub-polynomial complexity.}
While the analysis of the space complexity of parameterized problems dates back to the 1990s~\cite{cai_advice_1997}, a more in-depth analysis of parallel tractability started only recently with~\citet{elberfeld_space_2015}. Since then, it has been used for tighter analysis of problems, like \problem{MaxSAT}~\cite{bannach_parallel_2022}, and many tools, such as meta-theorems for Monadic Second-Order logic~\cite{bannach_parallel_2017}, have been developed for parallel tractability.

\paragraph{Organization.}
\Cref{sec:preliminaries} provides notation and preliminaries. \Cref{sec:courcelles} details upper bounds for treewidth and chordal graphs, and \Cref{sec:pw-interval} presents the dynamic programming approach for path decompositions and interval graphs. \Cref{sec:lower-bounds} establishes lower bounds, and \Cref{sec:conclusions} concludes.

\section{Preliminaries} \label{sec:preliminaries}

We adopt the graph theory notation of Bondy and Murty~\cite{BM2008}. A graph~$G$ is a pair~$(V(G), E(G))$, where~$v(G) = |V(G)|$ and~$e(G) = |E(G)|$. The neighborhood of a vertex~$v$ is~$N_G(v)$ and its degree is~$d_G(v)$, where~$d_G(v) = |N_G(v)|$.
We also use the usual interval notation~$[a,b] = \{a, a+1, \ldots, b\}$ for integers~$a$ and $b$ with $a \leq b$, and~$[a] = [1, a]$.

A graph~$G$ is an \emph{interval graph} if its vertices correspond to intervals on the real line with edges representing overlapping intervals. A graph is \emph{chordal} if it contains no induced cycles of length more than three. Note that every interval graph is chordal.

\paragraph{Sub-polynomial and parallel classes.} \label{def:complexity-classes}
Let~$\Sigma$ be a finite alphabet. A language~$Q \subseteq \Sigma^*$ is a \deff{decision problem}.
\deff{Logspace} (\L) is the class of problems decidable in $\Oh(\lg |x|)$ space, and \NL is its nondeterministic variant. 
For a real~$i \geq 1$, $\NC^i$ consists of problems decidable by a uniform family of \NC-circuits\footnote{A circuit family is uniform if there exists a logspace algorithm that, given~$1^{|x|}$, produces a circuit with input size~$|x|$.} (\mbox{\textup{\textsc{AND}}-}, \textup{\textsc{OR}}-, and \textup{\textsc{NOT}}-gates of bounded fan-in) with~$\Oh(\lg^i |x|)$ depth and polynomial size; $\NC = \bigcup_{i \geq 1} \NC^i$.
The prefix~\cclass{F} denotes the corresponding function class.

It is well-known that~$\NC^1 \subseteq \L \subseteq \NL \subseteq \NC^2 \subseteq \NC$~\cite[Theorem 16.1]{papadimitriou_computational_1994}. An (equivalent) alternative definition of~\NC is the class of problems that can be decided by a \PRAM (parallel random access machine) in polylogarithmic time on polynomially many processors. For simplicity, we consider only the \PRAM model with concurrent reads and writes (\CRCW-\PRAM). The following useful result allows us to translate results in the \PRAM model back into~\NC subclasses.

\begin{lemma}[{\citet[Theorem 2.56]{vollmer_introduction_1999}}] \label{lem:pram-to-nc}
    If a problem $Q$ admits a \CRCW-\PRAM solution running on input~$x$ in~$\Oh(\lg^i |x|)$ time on polynomially many processors, then~$Q \in \NC^{i+1}$.
\end{lemma}

\paragraph{Parameterized complexity.} \label{def:para-classes}
A \deff{parameterized problem} is a pair~$(Q, \kappa)$ of a language~$Q \subseteq \Sigma^*$ and a computable parameterization~$\kappa: \Sigma^* \to \N$. We denote parameterized problems with a leading ``$\p$-'', and the parameter as an index (e.g.,~\problem{$\p_{\tw}$-$k$-core}).
For a complexity class~\cclass{C}, a problem~$(Q, \kappa)$ is in the slice-wise class~\cclass{XC} if each slice~$Q_k = \{ x \in Q \mid \kappa(x) = k\}$ is in~\cclass{C}. It is in the para-class~\para{C} if there is a computable~$f$ and~$A \in \cclass{C}$ such that~$x \in Q \iff (x, 1^{f(\kappa(x))}) \in A$. Observe that~$\para{C} \subseteq \cclass{XC}$.

Note that~\para{P} is the well-known $\cclass{FPT}$ class, consisting of (parameterized) problems decidable in~$f(\kappa(x))\,|x|^{\Oh(1)}$ time, while~\cclass{XP} is the class of problems decidable in~$|x|^{\Oh(f(\kappa(x)))}$ time. Analogously,~\para{L} is the class of problems decidable using~$\Oh(f(\kappa(x))+\lg|x|)$ space, while problems in~\cclass{XL} may use~$\Oh(f(\kappa(x))\,\lg |x|)$ space. In the same manner, for a real~$i \geq 1$, $\paraNC^i$ problems are decidable by a uniform family of \NC-circuits with $\Oh(f(\kappa(x))+\lg^i |x|)$ depth and~$f(\kappa(x))\,|x|^{\Oh(1)}$ size. 
\paragraph{Degeneracy and~$k$-core.}

The~\deff{\kcore} of a graph~$G$ is its largest subgraph in which every vertex has degree at least~$k$. It is unique. The \deff{degeneracy} of~$G$ is the largest~$k$ for which the~\kcore of~$G$ is nonempty.  

In the \Kcore problem, given a graph~$G$, a vertex~$u \in V(G)$, and an integer~$k$, the goal is to decide if~$u$ is in the~\kcore of~$G$.

We add the C- or I- prefix when the problem is restricted to chordal or interval graphs, respectively. We add the~NE- prefix for the variant of deciding if the \kcore is \emph{nonempty}, which is a decision version of degeneracy, and F- for the \emph{function} variant of computing the~\kcore, that is, computing the word in~$\{0,1\}^{v(G)}$ that identifies the vertices in the~\kcore.

\paragraph{Monadic Second-Order logic (\MSO).}
We evaluate Monadic Second-Order (\MSO) formulas over a graph~$G$ with domain $V(G)$ and a binary relation $E(uv)$ representing edge adjacency. An \MSO formula is constructed from individual vertex variables (e.g., $u$, $v$), vertex set variables (e.g., $X$, $Y$), standard Boolean connectives $(\land, \lor, \neg, \Rightarrow, \Leftrightarrow)$, and quantifiers $(\exists$,~$\forall)$ over both individual vertices and sets of vertices. The allowable atomic formulas are equality $(u = v)$, set membership $(u \in X)$, and adjacency $(E(uv))$. We use~$G \vDash \varphi$ to denote that the graph~$G$ satisfies the formula~$\varphi$.

\newcommand{\T}{\ensuremath{\mathcal{T}}\xspace}
\paragraph{Tree and path decompositions.}
A \deff{tree decomposition} of~$G$ is a tree~\T whose nodes~$H \in V(\T)$ are subsets of~$V(G)$ such that \emph{(1)} each edge of~$G$ is contained in at least one node of~\T and \emph{(2)} the nodes containing each vertex form a subtree of~\T. We use \emph{vertex} for elements of~$V(G)$ and \emph{node} for elements of~$V(\T)$. The \deff{width} of~\T is~$\max\{|H|-1 \mid H \in V(\T)\}$, and the \deff{treewidth}~$\tw(G)$ is the minimum width over all tree decompositions of~$G$. When~\T is a path, it is a \deff{path decomposition}, and its minimum width is the \deff{pathwidth}~$\pw(G)$.

\begin{lemma}[\citet{bodlaender_contraction_2006}] \label{lem:degeneracy-leq-tw}
The degeneracy of a graph~$G$ is at most its treewidth. Equivalently, if~$k > \tw(G)$, then the~\kcore of~$G$ is empty.
\end{lemma}

The notation in the previous definitions may be combined, for example in~\problem{$\p_{\tw}$-F-$k$-core}, which is the function variant of \Kcore parameterized by treewidth. While~$k$ seems unbounded in that notation, by~\Cref{lem:degeneracy-leq-tw} any case where~$k > \tw(G)$ is trivial. Thus, we may assume $k \leq \tw(G)$, and a dual parameterization by treewidth and~$k$ (that is, \problem{$\p_{\tw,k}$-F-$k$-core}) is always redundant.

\section{Parameterization by treewidth, and chordal graphs} \label{sec:courcelles}

A famous meta-theorem of \citet{courcelle_graph_1990} states that if a problem is expressible in \MSO, it can be solved in linear time on graphs of bounded treewidth. This result has since been expanded in several ways for different models and bounds, such as the following lemma for parallel complexity.

\begin{lemma}[{\citet[Theorem 2]{bannach_parallel_2017}}] \label{lemma:mso-para-nc}
    If~\pA is a problem expressible in \MSO, then~$\p_{\tw,\varphi}\text{-\pA}$, the problem~\pA parameterized by treewidth and its formula size, can be efficiently solved in parallel. More precisely, \(\p_{\tw,\varphi}\text{-\pA} \in \paraNCeps.\)
\end{lemma}

Many problems, even several~$\NP$-complete ones, may be expressed in \MSO on graphs of bounded treewidth. For example, an \NP-complete variant of \Kcore, called the \emph{anchored \kcore problem}, is expressed in \MSO in~\cite[\S 9.2.2]{bulian_parameterized_2017}. We now express~\kcore in \MSO, and note that the size of the formula depends only on~$k$.
\begin{gather*}
    \textup{\textsc{D}}(a_1, \ldots, a_x): \bigand_{i \in [x-1]} \bigand_{j \in [i+1,x]} \left(a_i \neq a_j\right) \\[-1mm]
    \textup{\textsc{Is-$k$-set}}(K): \forall v \big(v \in K \Rightarrow \exists\,a_1 \cdots a_k\, \textup{\textsc{D}}(a_1, \ldots, a_k) \band \bigand_{i \in [k]}(a_i \in K \band E(va_i))\big) \\[-1mm]
    \textup{\textsc{In-$k$-core}}(u): \exists K ( u \in K \band \textup{\textsc{Is-$k$-set}}(K)) \\[3mm]
    \textup{\textsc{NE-$k$-core}}: \exists u\ \textup{\textsc{In-$k$-core}}(u)
\end{gather*}

It is straightforward to verify that~$G \vDash \textup{\textsc{In-\(k\)-core}}(u)$ if and only if~$u$ is in the~\kcore of~$G$, since~$\textup{\textsc{Is-$k$-set}}(K)$ ensures every vertex in~$K$ has degree at least~$k$ in~$K$.

\begin{theorem} \label{res:tw-kcore}
    \problem{$\p_{\tw}$-$k$-core} is in~\paraNCeps.
\end{theorem}

\begin{proof}
By~\Cref{lem:degeneracy-leq-tw}, we may assume~$k \leq \tw(G)$, and thus the formula \(\textup{\textsc{In-\(k\)-core}}(u)\) has bounded size for graphs of bounded treewidth. Finally, by~\Cref{lemma:mso-para-nc}, \(G \vDash \textup{\textsc{In-\(k\)-core}}(u)\) can be decided in~\paraNCeps, and thus we can determine if~$u$ is in the~\kcore of~$G$ in~\paraNCeps.\qed
\end{proof}

Several graph classes have bounded treewidth, including all~$\ell$-outerplanar graphs for bounded~$\ell$, series-parallel graphs, and Halin graphs.
Chordal graphs do not have bounded treewidth; however, we will show how to use~\Cref{res:tw-kcore} to solve \Kcore for a chordal graph when~$k$ is bounded.

\begin{definition}
    A \deff{deletion order for~\kcore} in a graph~$G$ is a sequence~$(u_1, \ldots, u_\ell)$ such that, for all~$i \in [\ell]$, the degree of~$u_i$ in~$G$ after removing vertices~$u_1, \ldots, u_{i-1}$ is less than~$k$.
\end{definition}

Note that a vertex~$u$ is contained in \emph{some} deletion order for~\kcore if and only if~$u$ is \emph{not} in the~\kcore of~$G$. The folklore linear algorithm for~\kcore consists of computing any maximal deletion order for~\kcore.

\begin{lemma} \label{res:chordal-g-to-gprime}
    Let~$G$ be a chordal graph,~$u \in V(G)$, and~$k \in \N$. Then there exists a subgraph~$G' \subseteq G$ with~$u \in V(G')$ and~$\tw(G') \leq k$ such that~$u$ is in the~\kcore of~$G$ if and only if~$u$ is in the~\kcore of~$G'$. Furthermore,~$G'$ can be computed in~$\cclass{FNC}^3$, and in~\cclass{FL} if~$G$ is an interval graph.
\end{lemma}

\begin{proof}
\newcommand{\Tt}{\ensuremath{\mathcal{T}'}\xspace}
    Since~$G$ is chordal, there exists a tree decomposition~\T of~$G$ such that each node of~\T is one of the~$\Oh(v(G))$ maximal cliques of~$G$~\cite[Theorem 4.1]{naor_fast_1989}.
    We say a node of~\T is \deff{heavy} if it contains at least~$k+1$ vertices of~$G$, and \deff{light} otherwise. A vertex of~$G$ is \deff{heavy} if it belongs to \emph{some} heavy node of~\T, and \deff{light} otherwise.
    If~$u$ is a heavy vertex, then it trivially belongs to the~\kcore of~$G$, and we can take~$G'$ as~$k+1$ vertices from a heavy node containing~$u$.
    Otherwise, denote by~\Tt the maximal subtree of~\T containing all nodes that contain~$u$ and not containing an \emph{internal} heavy node. Each heavy leaf~$H$ of~\Tt is adjacent to a light node~$L$. The vertices of~$H$ are either only in~$H$ or they must be in~$L$ as well (since the nodes containing them are connected in~\Tt). Note~$|H \cap L| \leq |L| \leq k$, so we can (arbitrarily) remove vertices from~$H \setminus L$ until~$H$ has exactly~$k+1$ vertices, and note that all light nodes are unaffected. Let~$G'$ be the subgraph induced by all remaining vertices in all nodes of~\Tt. Note that~$G'$ is chordal, because chordality is closed under induced subgraphs, and~\Tt is a tree decomposition of~$G'$. Thus, its largest clique has size at most~$k+1$, and~$\tw(G') = \omega(G') - 1 \leq k$ since~$G'$ is chordal. 

    We now argue that~$u$ is in the~\kcore of~$G$ if and only if~$u$ is in the~\kcore of~$G'$.
    Note that any light vertex~$v$ of~$G'$ satisfies~$N_{G'}(v) = N_G(v)$, that is, it has the same neighborhood in $G$ and~$G'$, and any vertex of~$G'$ has the same weight (heavy/light) in~$G$.
    If~$u$ is not in the~\kcore of~$G'$, then there exists a deletion order (for~\kcore)~$D$ in~$G'$ containing~$u$, and note that~$D$ is also a valid deletion order in~$G$; thus,~$u$ is not in the~\kcore of~$G$. Conversely, if~$u$ is not in the~\kcore of~$G$, then there exists a deletion order~$D$ in~$G$ containing~$u$ and, if we remove all vertices not in~$G'$ from~$D$, it becomes a valid deletion order for~$G'$ containing~$u$ since~$G' \subseteq G$; thus,~$u$ is not in the~\kcore of~$G'$.

    The known algorithms for computing~\T run in~$\Oh(\lg^2 v(G))$ time on a \PRAM~\cite{klein_efficient_1996,naor_fast_1989}; thus, by~\Cref{lem:pram-to-nc}, this can be done in~$\FNC^3$. However, if~$G$ is an interval graph, we can find its interval representation, and from that a path decomposition, in~\cclass{FL}~\cite{kobler_interval_2011}.
    
    Finally, finding the correct subtree~\Tt of \T uses reachability, which is in~\L~\cite{reingold_undirected_2008}, and removing vertices from the heavy leaves of~\Tt and computing~$G'$ can be easily done in~\cclass{FL}, which finishes the proof.\qed
    %
\undef\Tt
\undef\T
\end{proof}

\begin{theorem} \label{res:chordal-nc}
    \problem{$\p_k$-C-$k$-core} is in~$\paraNC^3$ and \problem{$\p_k$-I-$k$-core} in~\paraNCeps.
\end{theorem}

\begin{proof}
    Let~$G$ be a chordal graph. We can compute~$G'$ according to~\Cref{res:chordal-g-to-gprime} with a circuit of depth~$\Oh(\lg^3 v(G))$. By~\Cref{res:tw-kcore}, we can decide whether~$u$ is in the~\kcore of~$G'$ with a circuit of depth~$\Oh(f(k) + \lg^{2+\epsilon} v(G))$ for some computable~$f$. The final circuit has depth~$\Oh(f(k) + \lg^3 v(G))$, meaning the problem is in~$\paraNC^3$. The bound for interval graphs follows analogously by using~$\cclass{FL} \subseteq \cclass{FNC}^2$.\qed
\end{proof}

We immediately get the following.

\begin{corollary} \label{res:tw-chordal-fnc}
    \problem{$\p_{\tw}$-F-$k$-core} and \problem{$\p_k$-FI-$k$-core} are in~\para{\ensuremath{\FNC^{2+\epsilon}}}, and \problem{$\p_k$-FC-$k$-core} is in~\para{\ensuremath{\FNC^3}}.
\end{corollary}

\begin{proof}
    Recall these are the (function) problems of computing the~\kcore parameterized by treewidth in any graph, or by~$k$ in interval or chordal graphs, respectively.
    Use a copy of the circuit given by \Cref{res:tw-kcore,res:chordal-nc} for each vertex. The depth remains the same, and the size is still polynomial.\qed
\end{proof}

When we consider slice-wise tractability (X-classes instead of para-classes), another variant of Courcelle's result by~\citet{elberfeld_logspace_2010} shows that any problem~\pA expressible in \MSO on graphs of bounded treewidth can be solved in \L, which, in parameterized notation, means~$\p_{\tw,\varphi}\text{-\pA} \in \XL$. Note that \XL is not known to be comparable with \paraNCeps, so neither result is inherently stronger than the other. By this variant, we may easily obtain the following result.

\begin{corollary} \label{res:x-kcore}
    \problem{$\p_{\tw}$-F-$k$-core} and \problem{$\p_k$-FI-$k$-core} are in~\cclass{XFL}, and their decision variants \problem{$\p_{\tw}$-$k$-core} and \problem{$\p_k$-I-$k$-core} are in~\XL.
\end{corollary}

In summary, the results in this section show that, even though chordal graphs may have unbounded treewidth, large treewidth implies the existence of large cliques on which we can efficiently solve the problem, while still using Courcelle's results for the remaining ``tree-like'' graph.

\section{Parameterization by pathwidth, and interval graphs} \label{sec:pw-interval}

When given a path decomposition of~$G$, we can show that more efficient parallel algorithms exist for computing the~\kcore of~$G$. They depend on the width of the decomposition by a single exponential, which allows \NC algorithms even for logarithmic width. To do this, we cannot use \MSO formulas and Courcelle variants, as a tower of exponentials is inherent to that approach~\cite{lampis_first_2023}. Instead, we tailor-make a dynamic programming algorithm.

\renewcommand{\P}{\ensuremath{\mathcal{P}}\xspace}

Let~$G$ be a graph and \P be a path decomposition of~$G$. We can see the path decomposition as an assignment of an interval~$[\ell_u, r_u]$ to each vertex~$u$ such that if~$uv \in E(G)$ then the intervals of~$u$ and~$v$ touch. We may assume all interval endpoints are distinct integers between~$1$ and~$2v(G)$.
For each integer~$x$, let its \deff{stabbing set}~$C(x)$ be the set of vertices whose intervals contain~$x$, that is,~$\{v \in V(G) \mid x \in [\ell_v, r_v]\}$.

\newcommand{\K}{\ensuremath{\mathcal{K}}\xspace}

For each~$S \subseteq V(G)$, let the~\deff{$S$-anchored \kcore}~$\K(G, S)$ of~$G$ be its (unique) largest subgraph in which each vertex \emph{not in~$S$} has degree at least~$k$, that is, the~\kcore in which all vertices from~$S$ are \emph{forced}.\footnote{This is essentially the \kcore of a modified graph in which~$S$ is transformed into a clique, with additional vertices so that the clique has at least~$k+1$ vertices.} For any~$\ell$ and~$r$, let~$G[\ell,r]$ be the graph spanned by the intervals intersecting~$[\ell, r]$, that is, the subgraph induced by~$\{u \in V(G) \mid [\ell_u, r_u] \cap [\ell, r] \neq \emptyset\}$. 

\newcommand{\IK}{\textup{\textsc{InnerKCore}}}
\newcommand{\while}{\textbf{while}\xspace}

Let~$\ell, r \in [0, 2v(G)+1]$, $L \subseteq C(\ell)$, and $R \subseteq C(r)$. Define \[\IK(\ell, L, r, R) = \K(G[\ell, r], L \cup R),\] which is essentially the \kcore of~$G$ restricted to~$[\ell, r]$, while ``fixing'' the state of the vertices in the borders~$\ell$ and~$r$, independent of what the rest of~$G$ actually looks like. Note that the~\kcore of~$G$ is~$\IK(0, \emptyset, 2v(G)+1,\emptyset)$.

We aim to efficiently compute~$\IK(\ell,L,r,R)$. If~$r=\ell+1$, then we can directly compute~$\K(G[\ell,r],L\cup R)$ because the graph~$G[\ell,\ell+1]$ contains at most~$2(t+1)$ vertices, where~$t$ is the width of the path decomposition~\P. Otherwise, we pick any~$m \in [\ell+1,r-1]$ and then find an appropriate set~$M \subseteq C(m)$ such that~$\IK(\ell,L,m,M) \cup \IK(m,M,r,R)$ is the desired subgraph. Instead of considering all~$\Oh(2^t)$ possible values of~$M$, we can use the monotonicity of the~\kcore to consider at most~$t+2$ different values of~$M$, by essentially applying a peeling algorithm starting from~$M = C(m)$. This process is formalized as follows.

\begin{algorithm}[h]
    \caption{Compute the anchored \kcore of a graph~$G$ with a given path decomposition.}
    \label{alg:solve}
    \begin{algorithmic}[1]
        \Procedure{\IK}{$\ell, L, r, R$}
            \If{$r = \ell+1$}
                \State \Return $\K(G[\ell,r],L \cup R)$ \label{alg:solve:base-case}
            \EndIf
            \State $m \gets \floor{\frac{\ell+r}{2}}$
            \State $M \gets C(m)$
            \State $K \gets \IK(\ell, L, m, M) \cup \IK(m, M, r, R)$
            \While{$\{u \in M \setminus (L \cup R) \mid d_K(u) < k\} \neq \emptyset$}
                \State $M \gets M \setminus \{u \in M \setminus (L \cup R) \mid d_K(u) < k\}$ \label{alg:solve:del}
                \State $K \gets \IK(\ell, L, m, M) \cup \IK(m, M, r, R)$
            \EndWhile
            \State \Return $K$
        \EndProcedure
    \end{algorithmic}
\end{algorithm}

\begin{proposition}
    \Cref{alg:solve} correctly computes~$\IK(\ell, L, r, R)$.
\end{proposition}

\begin{proof}
    Let~$K^\star = \K(G[\ell, r], L \cup R)$. We want to show this is the value returned by the algorithm.
    The proof is by induction on~$r - \ell$. The base case, when~$r = \ell+1$, is trivial, since line~\ref{alg:solve:base-case} simply computes the correct value~$K^\star$.

    Now, let~$m = \floor{\frac{\ell+r}{2}}$ and note that~$\ell < m < r$. Let~$M \subseteq C(m)$ and
    \begin{align*}
        K(M) &= \IK(\ell, L, m, M) \cup \IK(m, M, r, R) \\
        &= \K(G[\ell,m],L \cup M) \cup \K(G[m, r], M \cup R),
    \end{align*}
    where the second equality follows from the induction hypothesis. We say~$M$ is \emph{valid} if~$d_{K(M)}(u) \geq k$ for all~$u \in M \setminus (L \cup R)$.
    Suppose~$M$ is invalid and~$d_{K(M)}(u)<k$ for some~$u \in M \setminus (L \cup R)$. Note \K is monotone, that is,~$\K(G, S') \subseteq \K(G,S)$ if~$S' \subseteq S$. Therefore,~$d_{K(M')}(u) \leq d_{K(M)}(u) < k$ for each~$M' \subseteq M$, which means no subset of~$M$ is valid if it contains~$u$. Since we start with~$M = C(m)$, which is trivially a superset of all valid sets, and at each iteration of the \while loop we only discard invalid sets by removing vertices from~$M$, eventually the algorithm finds~$M^\star$, the largest valid set. Since~$K^\star \cap C(m)$ is a valid set,~$K(M) \subseteq K^\star$ for every valid set~$M$, and~$\K$ is monotone, we have~$K(M^\star) = K^\star$, and the algorithm correctly computes this value.\qed
\end{proof}

\Cref{alg:solve} is correct but, if implemented naively, it may take~$\Omega(v(G))$ parallel time even if we run both subinstances $\IK(\ell, L, m, M)$ and $\IK(m, M, r, R)$ in parallel with an unbounded number of processors.\footnote{Take the example of a path and~$k=2$. Note that line~\ref{alg:solve:del} only deletes a vertex from~$M$ when its degree is less than~$k$. At any time, only the two endpoints of the path will have degree one, and after their deletion, the graph remains a path. So, there will be at least~$\nicefrac{v(G)}{2}$ sequential executions of line~\ref{alg:solve:del} throughout the execution of the algorithm. They may be in different processors, but they are still sequential.} We now show how to implement it efficiently.

\begin{theorem} \label{res:pw-kcore}
    Given~$G$ and a path decomposition~\P of~$G$ of width~$t$, the~\kcore of~$G$ can be computed in $\Oh(t \lg v(G))$ parallel time on~$\Oh((v(G) + t^2) \cdot v(G) \cdot 2^{2t})$ processors on a \PRAM.
\end{theorem}

\begin{proof}
    By~\Cref{lem:degeneracy-leq-tw}, the case where $k > t$ is trivial, and we may assume~$k \leq t$.

    From~\P we can easily compute intervals~$[\ell_u, r_u]$ for all vertices~$u$ of~$G$ so that if~$uv \in E(G)$ then~$[\ell_u,r_u] \cap [\ell_v,r_v] \neq \emptyset$, just by numbering the nodes of the path~\P from left to right. Furthermore, we may assume all interval endpoints are distinct integers in the~$[1,2v(G)]$ range by ordering and renumbering them in~$\Oh(\lg v(G))$ time on~$\Oh(v(G)^2)$ processors~\cite{jaja_introduction_1992}.

    Consider a directed acyclic graph (DAG) representing the recurrence implicit from~\Cref{alg:solve}. A node~$(\ell, L, r, R)$ may depend on (point to) several different nodes of type~$(\ell, L, m, M)$ and~$(m, M, r, R)$ for different values of~$M \subseteq C(m)$. The depth of this DAG, however, is~$\Theta(\lg v(G))$, because we pick~$m=\floor{\frac{\ell+r}{2}}$, roughly halving the size of the intervals~$[\ell, r]$. At depth~$i$, for a fixed~$M \subseteq C(m)$, all~$[\ell,r]$ intervals intersect only at their endpoints, totaling~$\Oh(v(G))$ intervals.
    Moreover,~$C(\ell)$ and~$C(r)$ each have at most~$2^{t+1}$ subsets; thus, at each depth there are~$\Oh(v(G) \cdot 2^{2t})$ nodes.
    
    Using a bottom-up dynamic programming approach, we can compute the values of~$\IK(\ell, L, r, R)$ in decreasing order of depth, since each execution of~\Cref{alg:solve} depends only on nodes of higher depth. At the highest depth, in line~\ref{alg:solve:base-case}, we can find the desired anchored~\kcore for a fixed node by repeatedly removing all non-forced vertices of degree less than~$k$, which takes~$\Oh(t \lg t)$ parallel time on~$\Oh(t^2 + v(G))$ processors since~$G[\ell,r]$ has at most~$2(t+1)$ vertices and~$\Oh(t^2)$ edges when~$r = \ell + 1$.
    
    For a fixed depth~$i$, assuming all values of~\IK\ for depth $i+1$ are computed, the rest of~\Cref{alg:solve} can be implemented in~$\Oh(t)$ time on~$\Oh(v(G))$ processors. For this, we implement the subgraph~$K$ as a list~$d_K \in \N^{V(G)}$ of the degree of each vertex of~$G$ in~$K$, and~$M$ as a list $\mathbbm{1}_M \in \{0,1\}^{C(m)}$ indicating which values from~$C(m)$ are still in~$M$. At each iteration of the \while loop, $d_K$ is computed in one (parallel) instruction adding up two lists from depth~$i+1$, and each vertex $u \in M$ determines if~$d_K(u) < k$, and, if so, sets~$\mathbbm{1}_M(u)$ to~0. If any vertex was removed from~$M$\footnote{This can be determined with one parallel concurrent write instruction.}, the \while continues.
    
    If we compute all nodes of the same depth in parallel on different processors, this takes~$\Oh(t)$ time on~$\Oh(v(G)^2 \cdot 2^{2t})$ processors per depth, and~$\Oh(t \lg t)$ time on~$\Oh((v(G) + t^2) \cdot v(G) \cdot 2^{2t})$ processors at the highest depth. This totals~$\Oh(t (\lg t + \lg v(G)))$ time on~$\Oh((v(G) + t^2) \cdot v(G) \cdot 2^{2t})$ processors.\qed
\end{proof}

Note there are algorithms in~\NC for obtaining a path decomposition of~$G$ when it has bounded pathwidth~\cite{bodlaender_parallel_1998}; however, they are not in~\NC when~$\pw(G) = \Theta(\lg v(G))$, which is why~\Cref{res:pw-kcore} requires a path decomposition as input.
On interval graphs, however, we can efficiently compute a path decomposition and, analogously to~\Cref{res:chordal-nc}, we can translate this result to interval graphs parameterized by~$k$ even though their pathwidth might be~$\omega(k)$.

\begin{corollary} \label{res:interval-kcore-full}
    Given an interval graph~$G$, its~\kcore can be computed in $\Oh((k + \lg v(G)) \lg v(G))$ parallel time on~$\Oh((v(G) + k^2) \cdot v(G) \cdot 2^{2k})$ processors on a \PRAM.
\end{corollary}

\begin{proof}
    Given an interval graph~$G$, start by computing an interval representation for~$G$ in $\Oh(\lg^2 v(G))$ time on~$\Oh(v(G)+e(G))$ processors~\cite{klein_efficient_1996}; that is, computing an interval~$[\ell_u,r_u]$ for each~$u \in V(G)$ such that two vertices are adjacent if and only if their intervals intersect. From this, in $\Oh(1)$ time on~$\Oh(v(G)^2)$ processors, we can easily compute a path decomposition~\P of~$G$ with nodes corresponding to~$C(\ell_u)$ for each~$u \in V(G)$.

    If a node in~\P contains at least~$k+1$ vertices, then they form a large clique that is trivially in the~\kcore, and we say that node is~\deff{heavy}. We then decompose~\P into all its maximal subpaths with no internal heavy node.\footnote{Each heavy node might be in two of the maximal paths, so duplicating them only at most doubles the size of the input.} For each such subpath~$\P'$, we remove vertices contained only in its endpoint nodes until they have at most~$k+1$ vertices. Each modified maximal path~$\P'$ has width~$k$ and we can use~\Cref{res:pw-kcore} to find the~\kcore of its induced subgraph~$G'$ in~$\Oh(k \lg v(G'))$ parallel time on~$\Oh((v(G') + k^2) \cdot v(G') \cdot 2^{2k})$ processors. By combining all~$k$-cores from each maximal path, together with all vertices in heavy nodes of~\P, we compute the~\kcore of~$G$ with the desired complexity.\qed
\end{proof}

\begin{corollary} \label{res:interval-nc}
    For each~$k = \Oh(\lg v(G))$, \problem{I-$k$-core} is in~$\NC^3$.
\end{corollary}

\begin{proof}
    When~$k \leq c \lg v(G)$ for some constant~$c$, by~\Cref{res:interval-kcore-full} there is an algorithm for~\problem{I-$k$-core} that takes~$\Oh(\lg^2 v(G))$ time and uses~$\Oh(v(G)^{2+2c})$ processors. Therefore, by~\Cref{lem:pram-to-nc} the problem is in~$\NC^3$.\qed
\end{proof}

This algorithm can also be generalized to chordal graphs via centroid decomposition, but remains in \NC only when~$k=\Oh(1)$, which is no stronger than the results from~\Cref{sec:courcelles}, so we omit it.

\renewcommand{\P}{\cclass{P}}

\section{Lower bounds} \label{sec:lower-bounds}

While~\Tcore is known to be~\P-complete in general graphs~\cite{anderson_p-complete_1984}, we just presented more efficient algorithms when restricted to interval graphs, chordal graphs, and graphs of bounded treewidth, that is, we provided upper bounds for this problem on such classes of graphs. We now complement these results by showing lower bounds for the problem on these same graph classes, which clarify the \emph{minimum} required computational resources to solve this problem even in these restricted classes. Consider the following problem.

In the \ORD problem, given a directed path~$P$ and vertices~$u$ and~$v$, the goal is to decide if there is a directed path from~$u$ to~$v$ in~$P$.

\citet{etessami_counting_1997} showed that \ORD is~\L-complete under $\NC^1$ reductions.\footnote{\citeauthor{etessami_counting_1997} actually showed \ORD is complete under \emph{quantifier-free projections}, which are a particular case of~$\NC^1$ reductions. However,~$\NC^1$ reducibility is enough for our results.} That is, if~\ORD can be solved in~$\NC^1$, then~$\L = \NC^1$, which is believed to be false. We extend this result to \Kcore on interval graphs of bounded pathwidth.

\begin{theorem} \label{res:kcore-l-hard-tw}
    For each~$k \geq 2$, \Kcore on interval graphs of pathwidth~$k$ is~\L-hard (under~$\NC^1$ reductions).
\end{theorem}

\begin{proof}
    We will show an~$\NC^1$ reduction from~\ORD to~\Kcore. Given a directed path~$P$ and vertices~$u$ and~$v$, let~$w$ be the vertex after~$u$ in~$P$. If~$u=v$, the answer to~\ORD is trivially ``yes'', and if~$u$ is the last vertex, it is ``no''. Create a new graph~$G_2$ starting from the undirected version of~$P$, and adding extra vertices~$u_1, u_2, v_1$, and $v_2$, and edges~$u_1u_2, uu_1, uu_2, v_1v_2, vv_1$, and $vv_2$. Essentially, we are turning~$u$ and~$v$ into vertices of triangles. See~\Cref{fig:ord-to-kcore}.

    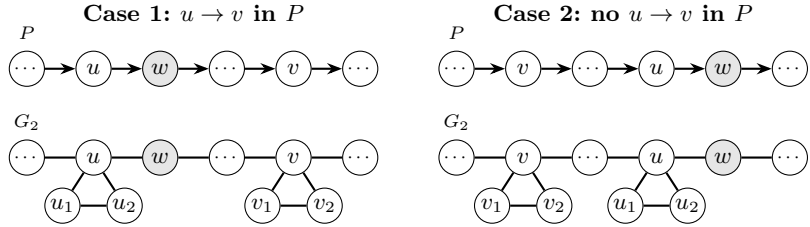
\begin{figure}[h]
        \centering
        \begin{tikzpicture}[
    every node/.style={font=\small},
    vtx/.style={circle,draw,inner sep=0.5pt,minimum size=13.5pt},
    hvtx/.style={vtx,fill=black!10},
    darrow/.style={-{Stealth[length=2.0mm]},thick},
    undedge/.style={thick},
    lab/.style={font=\scriptsize,fill=white,inner sep=1pt},
    scale=0.98, transform shape
]

\begin{scope}[shift={(-2.9,0)}]
    \node[font=\small\bfseries] at (0,1.95) {Case 1: $u \to v$ in $P$};

    \node[vtx] (l1) at (-2.25,1.2) {\scalebox{0.8}{$\cdots$}};
    \node[vtx] (u1) at (-1.35,1.2) {$u$};
    \node[hvtx] (w1) at (-0.45,1.2) {$w$};
    \node[vtx] (x1) at (0.45,1.2) {\scalebox{0.8}{$\cdots$}};
    \node[vtx] (v1) at (1.35,1.2) {$v$};
    \node[vtx] (r1) at (2.25,1.2) {\scalebox{0.8}{$\cdots$}};
    \path[every edge/.style={darrow, draw}] 
        (l1) edge (u1)
        (u1) edge (w1)
        (w1) edge (x1)
        (x1) edge (v1)
        (v1) edge (r1);
    \node[lab] at (-2.25,1.68) {$P$};

    \node[lab] at (-2.25,0.48) {$G_2$};
    \node[vtx] (l2) at (-2.25,0.0) {\scalebox{0.8}{$\cdots$}};
    \node[vtx] (u2) at (-1.35,0.0) {$u$};
    \node[hvtx] (w2) at (-0.45,0.0) {$w$};
    \node[vtx] (x2) at (0.45,0.0) {\scalebox{0.8}{$\cdots$}};
    \node[vtx] (v2) at (1.35,0.0) {$v$};
    \node[vtx] (r2) at (2.25,0.0) {\scalebox{0.8}{$\cdots$}};
    \draw[undedge] (l2) -- (u2) -- (w2) -- (x2) -- (v2) -- (r2);

    \node[vtx] (uu1) at ([shift={(-0.42,-0.65)}]u2) {$u_1$};
    \node[vtx] (uu2) at ([shift={(0.42,-0.65)}]u2) {$u_2$};
    \draw[undedge] (u2) -- (uu1) -- (uu2) -- (u2);

    \node[vtx] (vv1) at ([shift={(-0.42,-0.65)}]v2) {$v_1$};
    \node[vtx] (vv2) at ([shift={(0.42,-0.65)}]v2) {$v_2$};
    \draw[undedge] (v2) -- (vv1) -- (vv2) -- (v2);
\end{scope}

\begin{scope}[shift={(2.9,0)}]
    \node[font=\small\bfseries] at (0,1.95) {Case 2: no $u \to v$ in $P$};

    \node[vtx] (l1b) at (-2.25,1.2) {\scalebox{0.8}{$\cdots$}};
    \node[vtx] (v1b) at (-1.35,1.2) {$v$};
    \node[vtx] (y1b) at (-0.45,1.2) {\scalebox{0.8}{$\cdots$}};
    \node[vtx] (u1b) at (0.45,1.2) {$u$};
    \node[hvtx] (w1b) at (1.35,1.2) {$w$};
    \node[vtx] (z1b) at (2.25,1.2) {\scalebox{0.8}{$\cdots$}};
    \path[every edge/.style={darrow, draw}] 
        (l1b) edge (v1b)
        (v1b) edge (y1b)
        (y1b) edge (u1b)
        (u1b) edge (w1b)
        (w1b) edge (z1b);
    \node[lab] at (-2.25,1.68) {$P$};

    \node[lab] at (-2.25,0.48) {$G_2$};
    \node[vtx] (l2b) at (-2.25,0.0) {\scalebox{0.8}{$\cdots$}};
    \node[vtx] (v2b) at (-1.35,0.0) {$v$};
    \node[vtx] (y2b) at (-0.45,0.0) {\scalebox{0.8}{$\cdots$}};
    \node[vtx] (u2b) at (0.45,0.0) {$u$};
    \node[hvtx] (w2b) at (1.35,0.0) {$w$};
    \node[vtx] (z2b) at (2.25,0.0) {\scalebox{0.8}{$\cdots$}};
    \draw[undedge] (l2b) -- (v2b) -- (y2b) -- (u2b) -- (w2b) -- (z2b);

    \node[vtx] (uu1b) at ([shift={(-0.42,-0.65)}]u2b) {$u_1$};
    \node[vtx] (uu2b) at ([shift={(0.42,-0.65)}]u2b) {$u_2$};
    \draw[undedge] (u2b) -- (uu1b) -- (uu2b) -- (u2b);

    \node[vtx] (vv1b) at ([shift={(-0.42,-0.65)}]v2b) {$v_1$};
    \node[vtx] (vv2b) at ([shift={(0.42,-0.65)}]v2b) {$v_2$};
    \draw[undedge] (v2b) -- (vv1b) -- (vv2b) -- (v2b);
\end{scope}

\end{tikzpicture}
        \caption{Reduction in \Cref{res:kcore-l-hard-tw}. Vertex $w$ (shaded) is in the $2$-core of $G_2$ exactly in Case~1, that is, when there is a directed path from $u$ to $v$ in $P$.}
        \label{fig:ord-to-kcore}
    \end{figure}

    Note that the only cycles in~$G_2$ are these two triangles, so~$w$ is in the~$2$-core of~$G_2$ if and only if it has internally disjoint paths to the two triangles. That is the case only if~$w$ is in the (undirected) path between~$u$ and~$v$, which corresponds in~$P$ either to the directed path from~$u$ to~$v$, which would visit~$w$, or from~$v$ to~$u$, which would not. Thus,~$w$ is in the~$2$-core of~$G_2$ if and only if~$u$ has a directed path to~$v$ in~$P$.

    Note~$G_2$ is an interval graph of pathwidth~$2$. Now create~$G_k$ from~$G_2$ by adding~$k-2$ new vertices, each connected to all the other vertices of~$G_k$. Hence,~$w$ is in the~\kcore of~$G_k$ if and only if~$u$ has a directed path to~$v$ in~$P$. The graph~$G_k$ remains an interval graph with pathwidth~$k$.\footnote{Informally, you may consider each new vertex as a new interval in~$G_k$ intersecting all the other intervals, and for a path decomposition of~$G_k$ you can add each new vertex to all nodes of a path decomposition of~$G_2$.} This reduction from~$(P,u,v)$ to~$(G_k,w)$ is trivially in~$\NC^1$, finishing the proof.\qed
\end{proof}

\begin{corollary} \label{res:kcore-para-l-hard}
    \problem{$\p_{\tw}$-$k$-core},~\problem{$\p_k$-C-$k$-core}, and~\problem{$\p_k$-I-$k$-core} are~\para{L}-hard (under $\paraNC^1$ reductions).
\end{corollary}

\section{Conclusion}\label{sec:conclusions}

We have shown that~\Kcore is in~$\paraNCeps$ (and \XL) when parameterized by treewidth, and in $\paraNC^3$ when parameterized by~$k$ on chordal graphs. We also showed an algorithm for~\Kcore which is in~$\NC^3$ on interval graphs when~$k = \Oh(\lg v(G))$, or on graphs of pathwidth~$\Oh(\lg v(G))$ when given a path decomposition. The results are summarized in~\Cref{tab:summary}. To improve the upper bounds, it seems necessary to improve (or bypass) the tree decomposition algorithms, as they are the current bottleneck.

Our results highlight the contrast between \Kcore and \NEkcore. In a chordal graph~$G$, \problem{NEC-$k$-core} is in~\NC because the~\kcore is nonempty if and only if~$k < \omega(G)$, which can be computed efficiently. However, identifying all vertices in the~\kcore remains open for chordal graphs with~$k = \omega(1)$ and interval graphs with~$k = \omega(\lg v(G))$.


Another open area is planar graphs, where degeneracy is at most~$5$. While~\Cref{res:tw-kcore} provides an~\NC algorithm on~$\ell$-outerplanar graphs with bounded~$\ell$, the complexity remains open on general planar graphs for $3 \leq k \leq 5$. Existing \P-hardness reductions~\cite{anderson_p-complete_1984,couto_hardness_2025} do not carry over to planar graphs, and planarizing gadgets cannot be used~\cite{couto_hardness_2025_full}.

Finally, several variants such as approximations, directed versions, and~$k$-truss are~\P-complete~\cite{couto_hardness_2025,anderson_p-complete_1984}, and parameterized algorithms similar to the ones presented in this paper are likely possible.

\bibliography{references}
\end{document}